\documentclass[sigconf]{acmart}

\AtBeginDocument{%
  \providecommand\BibTeX{{%
    \normalfont B\kern-0.5em{\scshape i\kern-0.25em b}\kern-0.8em\TeX}}}

\setcopyright{acmcopyright}
\copyrightyear{2018}
\acmYear{2018}
\acmDOI{XXXXXXX.XXXXXXX}

\usepackage{color}

\usepackage[utf8]{inputenc} 
\usepackage[T1]{fontenc}    
\usepackage{hyperref}       
\usepackage{url}            
\usepackage{booktabs}       
\usepackage{amsfonts}       
\usepackage{nicefrac}       
\usepackage{microtype}      
\usepackage{xcolor}         
\usepackage{float}
\usepackage{booktabs}
\usepackage{array}
\usepackage{balance} 
\usepackage{multirow}
\usepackage[normalem]{ulem}
\usepackage{color}
\definecolor{lightgray}{RGB}{215,215,215}
\usepackage{colortbl}  
\usepackage{xcolor}
\useunder{\uline}{\ul}{}
\usepackage{subfigure}
\usepackage{amsmath}
\usepackage{wrapfig}

\newtheorem{theorem}{Theorem}
\newtheorem{lemma}{Lemma}
\usepackage{algorithm}  
\usepackage{algorithmicx}  
\usepackage{algpseudocode}  
\usepackage{amsmath}  
\usepackage{enumitem}

\floatname{algorithm}{Algorithm}

\usepackage{amsmath}  
\usepackage{enumitem}
\usepackage{tabularx}
\usepackage[utf8]{inputenc}
\usepackage[english]{babel}
\usepackage{amsthm}
\usepackage{bm}
\usepackage{graphicx}
\usepackage{caption}
\usepackage{wrapfig}
\usepackage{verbatim}

\newcommand{\ie}{\emph{i.e., }}
\newcommand{\eg}{\emph{e.g., }}

\newcommand{\cf}{\emph{cf. }}

\newtheorem{proposition}{Proposition}

\copyrightyear{2024}
\acmYear{2024}
\setcopyright{acmlicensed}\acmConference[CIKM '24]{Proceedings of the 33rd ACM International Conference on Information and Knowledge Management}{October 21--25, 2024}{Boise, ID, USA}
\acmBooktitle{Proceedings of the 33rd ACM International Conference on Information and Knowledge Management (CIKM '24), October 21--25, 2024, Boise, ID, USA}
\acmDOI{10.1145/3627673.3679569}
\acmISBN{979-8-4007-0436-9/24/10}
\begin{document}

\title{Learnable Item Tokenization for Generative Recommendation}

\author{Wenjie Wang}
\authornote{denotes equal contribution. This research is supported by NExT Research Center.}
\email{wenjiewang96@gmail.com}
\affiliation{
\institution{National University of Singapore}
\country{Singapore}
}

\author{Honghui Bao}
\authornotemark[1]
\email{honghuibao2000@gmail.com}
\affiliation{
  \institution{National University of Singapore}
\country{Singapore}
}
\author{Xinyu Lin}
\email{xylin1028@gmail.com}
\affiliation{
\institution{National University of Singapore}
\country{Singapore}
}

\author{Jizhi Zhang}
\email{cdzhangjizhi@mail.ustc.edu.cn}
\affiliation{
\institution{University of Science and Technology of China}
\city{Hefei}
\country{China}
}
\author{Yongqi Li}
\email{liyongqi0@gmail.com}
\affiliation{
\institution{The Hong Kong Polytechnic University}
\city{Hong Kong SAR}
\country{China}
}

\author{Fuli Feng}
\email{fulifeng93@gmail.com}
\authornote{Corresponding author.}
\affiliation{
\institution{University of Science and Technology of China}
\city{Hefei}
\country{China}
}

\author{See-Kiong Ng}
\email{seekiong@nus.edu.sg}
\affiliation{
  \institution{National University of Singapore}
\country{Singapore}
}

\author{Tat-Seng Chua}
\email{dcscts@nus.edu.sg}
\affiliation{
\institution{National University of Singapore}
\city{}
\country{Singapore}
}
\renewcommand{\shortauthors}{Wenjie Wang et al.}
\begin{abstract}

Utilizing powerful Large Language Models (LLMs) for generative recommendation has attracted much attention. 
Nevertheless, a crucial challenge is transforming recommendation data into the language space of LLMs through effective item tokenization. 
Current approaches, such as ID, textual, and codebook-based identifiers, exhibit shortcomings in encoding semantic information, incorporating collaborative signals, or handling code assignment bias.
To address these limitations, we propose LETTER (a LEarnable Tokenizer for generaTivE Recommendation), which integrates hierarchical semantics, collaborative signals, and code assignment diversity to satisfy the essential requirements of identifiers. 
LETTER incorporates Residual Quantized VAE for semantic regularization, a contrastive alignment loss for collaborative regularization, and a diversity loss to mitigate code assignment bias. We instantiate LETTER on two models and propose a ranking-guided generation loss to augment their ranking ability theoretically. Experiments on three datasets validate the superiority of LETTER, advancing the state-of-the-art in the field of LLM-based generative recommendation. 
\vspace{-0.25cm}
\end{abstract}





\maketitle

\section{Introduction}
\label{sec:introduction}

Harnessing Large Language Models (LLMs) for generative recommendation has emerged as a prominent avenue, drawing substantial research interest~\cite{hua2023index,rajput2023recommender}. 
Leveraging the powerful abilities of LLMs (\eg rich world knowledge, reasoning, and generalization), numerous efforts are investigating LLMs for generative recommendation, showcasing the potential of building foundational LLMs to move beyond natural language processing to information retrieval and recommendation domains~\cite{geng2022recommendation,cui2022m6}. 
As illustrated in Figure~\ref{fig:overview}, given user historically interacted items, generative recommendation aims to utilize LLMs to generate target items as recommendations. 
A key problem to encoding and generating items via LLMs is \textit{item tokenization}, which indexes each item via an identifier (\ie a token sequence), thereby bridging the gap between recommendation data and the language space of LLMs. 
As such, it is imperative to investigate item tokenization for LLM-based generative recommendation.


\begin{figure}[t]
\setlength{\abovecaptionskip}{0cm}
\setlength{\belowcaptionskip}{-0.30cm}
\centering
\includegraphics[scale=0.58]{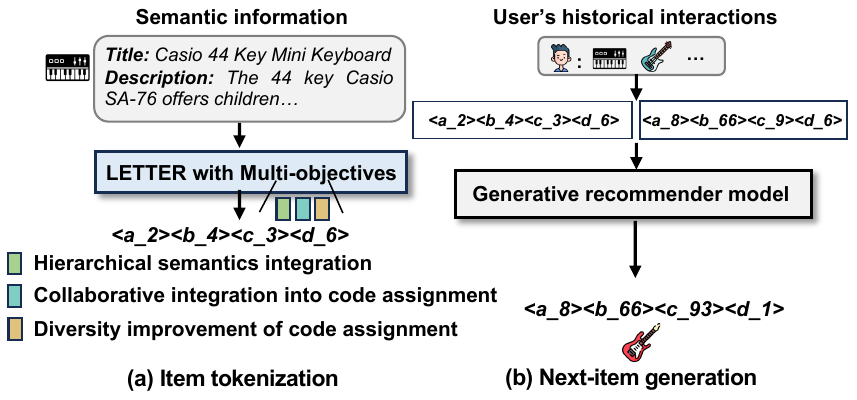}
\caption{Overview of utilizing LETTER for LLM-based generative recommendation.}
\label{fig:overview}
\end{figure}

Existing item tokenization approaches have mainly explored three types of identifiers to index an item: 
\begin{itemize}[leftmargin=*]
    \item ID identifiers assign each item with a unique numerical string (\eg ``28,128''), ensuring the identifier uniqueness~\cite{geng2022recommendation,hua2023index}. However, numerical IDs are inefficient in encoding semantic information, making it challenging to generalize to cold-start items~\cite{tan2024towards}. 
    
    \item Textual identifiers directly leverage item semantic information such as titles, attributes, or/and descriptions as item identifiers~\cite{bao2023bi,hua2023index}. Nevertheless, such textual identifiers suffer from the following issues: 
    1) semantic information is not hierarchically distributed from coarse to fine-grained in token sequence of the identifiers. This results in low generation probability if the first tokens of the target item are not closely aligned with user preference (\eg ``With'' and ``Be'' in movie titles); 
    2) textual identifiers lack collaborative signals from user behaviors~\cite{bao2023bi}.
    The token sequence of an identifier is solely determined by semantic information. As such, items with similar semantics yet differing collaborative signals exhibit similar token sequences, undermining the distinctiveness of items from the perspective of Collaborative Filtering (CF) (refer to Section~\ref{sec:task} and Figure~\ref{fig:cf_misalignment}). 

    \item Codebook-based methods adopt auto-encoders to encode item semantics into hierarchical code sequences as identifiers~\cite{rajput2023recommender, zheng2023adapting}. However, similar to textual identifiers, codebook-based identifiers still lack collaborative signals in code sequences (see Section~\ref{sec:task}). 
    Moreover, as shown in Figure~\ref{fig:Intro-distribution}, the assignment of codes to items is naturally imbalanced, causing the item generation bias where items with high-frequency codes are more readily generated.      
\end{itemize}




In light of these considerations, an ideal identifier should meet the following criteria: 
\begin{itemize}[leftmargin=*]
    \item Integrating hierarchical semantics into identifiers, where the token sequence initially encodes broad and coarse-grained semantics, progressively transitioning into more refined, fine-grained details. This aligns with the autoregressive generation characteristics of generative recommendation. 
    \item Incorporating collaborative signals into the token assignment, which ensures that items with similar collaborative signals in user behaviors possess similar token sequences as identifiers.
    \item Improving the diversity of token assignments to alleviate the item generation bias, thereby ensuring fairness in item generation.
\end{itemize}


To this end, we propose a \textbf{LE}arnable \textbf{T}okenizer for genera\textbf{T}iv\textbf{E} \textbf{R}ecommendation named LETTER. 
LETTER incorporates three kinds of regularization to enhance the codebook-based identifiers. In particular, semantic regularization first integrates Residual Quantized VAE (RQ-VAE)~\cite{lee2022autoregressive} to translate the item semantic information into the hierarchical identifiers~\cite{rajput2023recommender}. Built on this, collaborative regularization utilizes a contrastive alignment loss to align semantic quantized embeddings in RQ-VAE with CF embeddings from a well-trained CF model (\eg LightGCN~\cite{he2020lightgcn}). Moreover, LETTER introduces a diversity loss to enhance code embedding diversity to alleviate code assignment bias and item generation bias. We instantiate LETTER on two representative generative recommender models and propose a ranking-guided generation loss to boost the ranking ability of these generative models theoretically. Extensive experiments on three datasets with in-depth investigation have validated that LETTER can achieve superior item tokenization by simultaneously considering hierarchical semantics, collaborative signals, and code assignment diversity in identifiers. For reproducibility, we release our code and data at \url{https://github.com/HonghuiBao2000/LETTER}.

To sum up, the contributions of this work are as follows.
\begin{itemize}[leftmargin=*]
    \item We comprehensively analyze the necessary features of an ideal identifier. Accordingly, we propose a novel learnable tokenizer named LETTER, aimed at adaptively learning identifiers that encompass hierarchical semantics, collaborative signals, and code assignment diversity.

    \item We instantiate LETTER on two generative recommender models and leverage a ranking-guided generation loss to theoretically enhance the ranking ability of generative recommender models.

    \item We conduct extensive experiments on three datasets, coupled with the in-depth investigation with diverse settings, validating that LETTER outperforms existing item tokenization methods for generative recommendation.
    
\end{itemize}
\section{Item Tokenization}
\label{sec:task}

Generative recommendation aims to generate the next item given the user's historical interactions~\cite{li2023large}. 
A crucial fundamental step in LLM-based generative recommendation lies in item tokenization, which assigns an identifier to each item. 
Each identifier is a sequence of tokens to assist LLMs with item encoding and generation. 
By item tokenization, LLM-based generative recommender models can 1) transform the user's historical interactions into a sequence of identifiers for item encoding and 2) autoregressively generate an item identifier for next-item recommendation (see Figure~\ref{fig:overview}(b)). 

\begin{figure}[t]
\setlength{\abovecaptionskip}{0cm}
\setlength{\belowcaptionskip}{-0.30cm}
\centering
\includegraphics[scale=0.66]{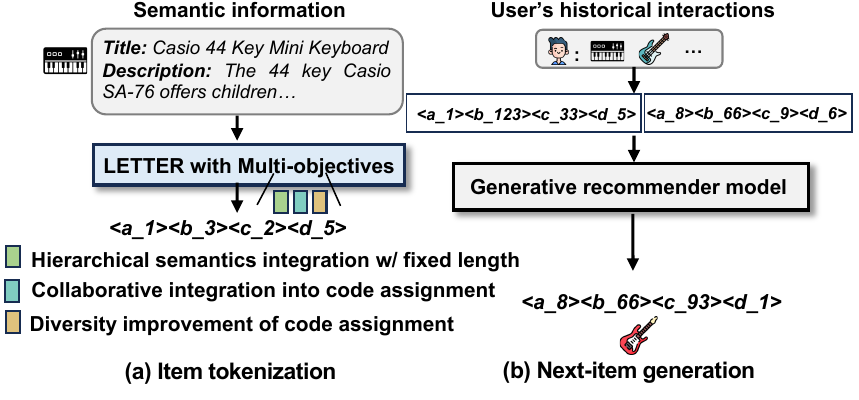}
\caption{Misalignment between item identifiers and collaborative signals. ``Emb.'' denotes ``Embeddings''.}
\label{fig:cf_misalignment}
\end{figure}

\begin{figure}[t]
\setlength{\abovecaptionskip}{0.01cm}
\setlength{\belowcaptionskip}{-0.50cm}
\centering
\includegraphics[scale=0.42]{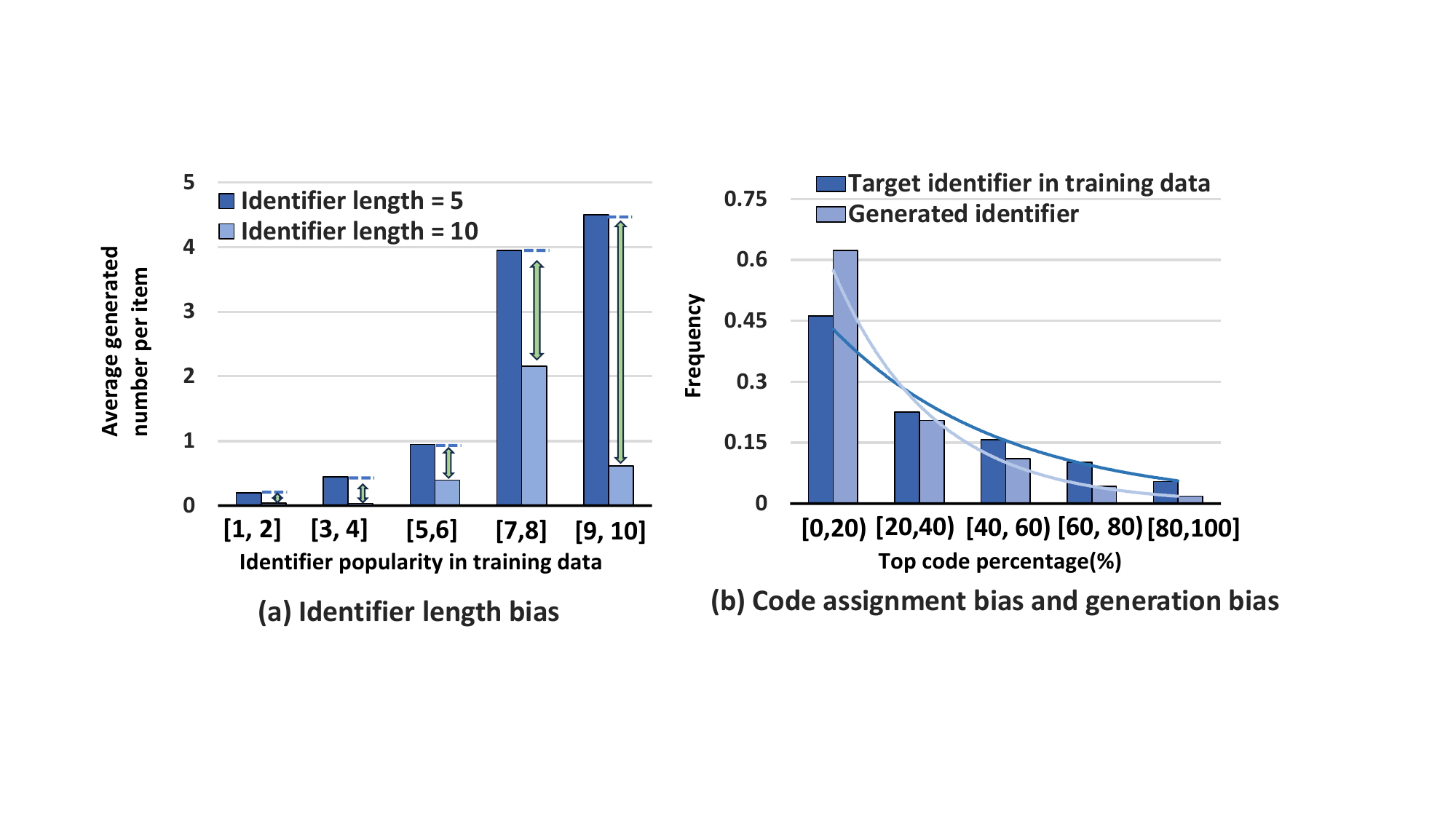}
\caption{Illustration of code assignment bias and generation bias on Instruments.}
\label{fig:Intro-distribution}
\end{figure}

\begin{figure*}[t]
\setlength{\abovecaptionskip}{0.10cm}
\setlength{\belowcaptionskip}{-0.25cm}
\centering
\includegraphics[scale=0.60]{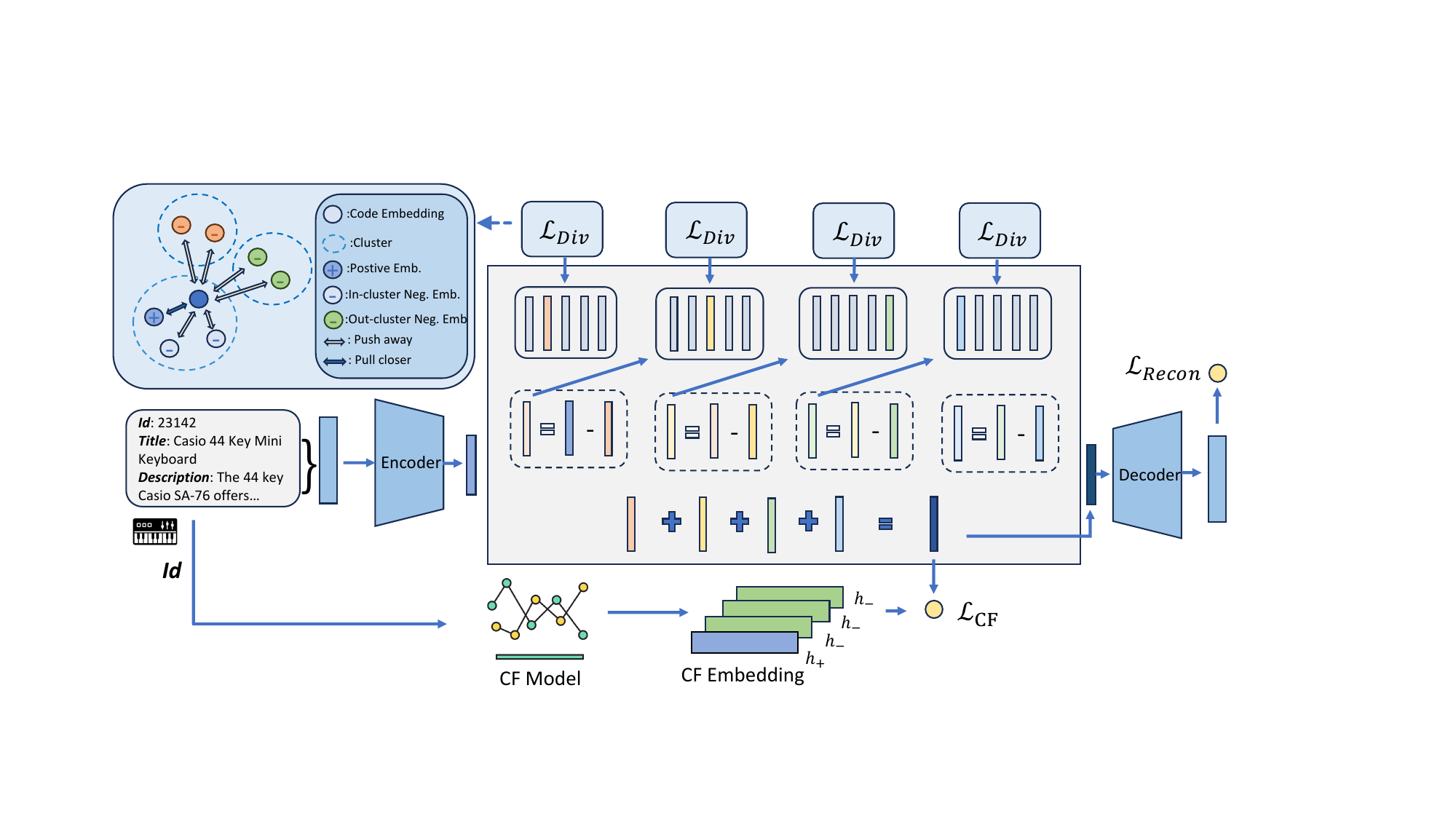}
\caption{Illustration of LETTER with three kinds of regularization, where semantic regularization ensures the semantic encoding, collaborative regularization enhances the alignment between the identifiers' code sequence and collaborative signals, and diversity regularization alleviates the code assignment bias.}
\label{fig:Main-Figure}
\end{figure*}
\vspace{3pt}

Existing item token methods have explored ID identifiers~\cite{hua2023index, geng2022recommendation, wang2024enhanced,chu2023leveraging}, textual identifiers~\cite{bao2023bi,zhang2021language,zhang2023recommendation}, and codebook-based identifiers~\cite{rajput2023recommender,zheng2023adapting}. However, ID identifiers have insufficient capability in semantic encoding~\cite{tan2024towards}. Moreover, we present more details on the critical issues in textual and codebook-based identifiers. 
\begin{itemize}[leftmargin=*]

    \item \textbf{Non-hierarchical semantics.} As mentioned in Section~\ref{sec:introduction}, semantic information is not hierarchically distributed from coarse to fine-grained in textual identifiers. 

    \item \textbf{Lack of collaborative signals.} Previous textual and codebook-based identifiers do not consider collaborative signals in the composition of the token sequence. They typically inject collaborative signals into the token embeddings of identifiers during the LLM training stage (\eg TIGER~\cite{rajput2023recommender} and LC-Rec~\cite{zheng2023adapting}). However, as illustrated in Figure~\ref{fig:cf_misalignment}, two items with similar semantics always share similar token sequences and embeddings, making it challenging to align with collaborative signals. Injecting collaborative signals into identifier embeddings of these two items by training will cause collisions. 

    \item \textbf{Code assignment bias.} Codebook-based methods additionally suffer from imbalanced code assignment, leading to item generation bias. We utilize TIGER~\cite{rajput2023recommender} on Instruments dataset for empirical investigation. Specifically, we split item identifier codes into five groups based on their assignment popularity in the training data and obtain the generation frequency of TIGER. From Figure~\ref{fig:Intro-distribution}, we can observe that popular item identifiers are more likely to be generated, amplifying the generation bias. 

\end{itemize}


In light of these, we consider three essential objectives to pursue an ideal identifier: 
1) hierarchical semantic integration with the token sequence, 
2) collaborative signals integration in the token sequence, and 
3) high diversity of the code assignment.


\section{Method}
\label{sec:method}


We elaborate on LETTER in Section~\ref{sec:tokenizer} and apply LETTER to LLM-based generative recommender models in Section~\ref{sec:instantiation}.

\subsection{LETTER}\label{sec:tokenizer}
To achieve the three objectives for item tokenization, we build LETTER for LLM-based generative recommendation. 
In particular, as depicted in Figure~\ref{fig:Main-Figure}, LETTER can generate 
identifiers, \ie code sequences, with hierarchical semantics. 
Moreover, 
we introduce collaborative regularization to endow the code sequence with collaborative signals, and diversity regularization to alleviate the code assignment bias issue during tokenizer training. 

\vspace{2pt}
\noindent$\bullet\quad$\textbf{Semantic regularization}. 
%
To pursue identifiers with hierarchical semantics,
 we develop LETTER based on RQ-VAE~\cite{lee2022autoregressive}, which encodes the hierarchical item semantics into a code sequence. 
As a multi-level embedding quantizer, RQ-VAE recursively quantizes semantic residuals using a fixed number of codebooks, thus naturally achieving  
identifier with coarse to fine-grained semantics. 
As shown in Figure~\ref{fig:Main-Figure}, the learnable tokenizer generates the hierarchical semantic identifier via two steps: 1) semantic embedding extraction, and 2) semantic embedding quantization. 

\textbf{\textit{Semantic embedding extraction}}. Given an item with its content information such as titles and descriptions, we first extract the semantic embedding $\bm{s}$ through a pre-trained semantic extractor, \eg LLaMA-7B~\cite{touvron2023llama}. 
The semantic embedding is then compressed into the latent semantic embedding $\bm{z}\in\mathbb{R}^{d}$ through an encoder $\bm{z} = \text{Encoder}(\bm{s})$.  

\textbf{\textit{Semantic embedding quantization}}. 
The latent semantic embedding $\bm{z}$ is then quantized into the code sequence through $L$-level codebooks, where $L$ is the identifier length. 
Specifically, for each code level $l\in\{1,\dots,L\}$, we have a codebook $\mathcal{C}_l=\{\bm{e}_i\}_{i=1}^N$, where $\bm{e}_i\in\mathbb{R}^{d}$ is a learnable code embedding and $N$ denotes the codebook size. 
Subsequently, the residual quantization can be formulated as 
\begin{equation}
\left\{
\begin{aligned}
    &c_l=\mathop{\arg\min}_{i}\|\bm{r}_{l-1}-\bm{e}_i\|^2, \quad \bm{e}_i\in\mathcal{C}_l, \\
    &\bm{r}_l=\bm{r}_{l-1}-\bm{e}_{c_l},
\end{aligned}
\right.
\end{equation}
where $c_l$ is the assigned code index from the $l$-th level codebook, $\bm{r}_{l-1}$ is the semantic residual from the last level, and we set $\bm{r}_0=\bm{z}$. 
Intuitively, at each code level, LETTER finds the most similar code embedding with the semantic residual and assigns the item with the corresponding code index. 
After the recursive quantization, we eventually obtain the quantized identifier $\Tilde{i}=[c_{1}, c_{2}, \dots, c_{L}]$, and the quantized embedding $\hat{\bm{z}}=\sum_{l=1}^{L}\bm{e}_{c_l}$. 
The quantized embedding $\hat{\bm{z}}$ is then decoded to the reconstructed semantic embedding $\hat{\bm{s}}$. 
The loss for semantic regularization is formulated as:
\begin{equation}
\left\{
\begin{aligned}
    &\mathcal{L}_{\text{Sem}}=\mathcal{L}_{\text{Recon}}+\mathcal{L}_{\text{RQ-VAE}}, \quad \text{where,} \\
    &\mathcal{L}_{\text{Recon}}=\|\bm{s}-\hat{\bm{s}}\|^{2}, \\
    &\mathcal{L}_{\text{RQ-VAE}}=\sum_{l=1}^{L}\|\text{sg}[\bm{r}_{l-1}]-\bm{e}_{c_l}\|^2 + \mu \|\bm{r}_{l-1}-\text{sg}[\bm{e}_{c_{l}}]\|^2,
\end{aligned}
\right.
\end{equation}
where $\text{sg}[\cdot]$ is the stop-gradient operation~\cite{van2017neural}, and $\mu$ is the coefficient to balance the strength between the optimization of code embeddings and encoder. 
Specifically, the reconstruction loss $\mathcal{L}_{\text{Recon}}$ aims to maintain the essential semantics information in the latent space such that the semantic embedding can be reconstructed from the quantized embedding. 
Besides, $\mathcal{L}_{\text{RQ-VAE}}$ reduces the residual error for all levels and jointly trains the encoder, and code embeddings~\cite{van2017neural}. 
By semantic regularization, the code sequence encodes hierarchical semantics, facilitating coarse-grained to fine-grained generation and cold-start generalization. 

\vspace{3pt}
\noindent$\bullet\quad$\textbf{Collaborative regularization}. 
%
%
To inject collaborative signals into the code sequence instead of only code embeddings, we introduce collaborative regularization to align the quantized embedding $\hat{\bm{z}}$ and the CF embedding via contrastive learning. 
Specifically, we utilize a well-trained CF recommender model (\eg SASRec~\cite{kang2018self} and LightGCN~\cite{he2020lightgcn}) to obtain CF embeddings of items, which is then aligned with the quantized embedding $\hat{\bm{z}}$ through a CF loss: 
\begin{equation}
    \mathcal{L}_{\text{CF}} = -\frac{1}{B}\sum_{i=1}^{B} \frac{\exp (<\hat{\bm{z}_i},\bm{h}_i>)}{\sum_{j=1}^{B} \exp (<\hat{\bm{z}}_i,\bm{h}_j>)},
\end{equation}
where $\bm{h}$ refers to the CF embedding, $<\cdot ,\cdot>$ denotes the inner product, and $B$ is the batch size. 

Intuitively, collaborative regularization encourages items with similar collaborative interactions to exhibit similar code sequences. 
In contrast to the typical codebook-based tokenization approaches such as TIGER~\cite{rajput2023recommender} that relies solely on semantics to generate identifiers, we inject collaborative signals into the code embeddings by optimizing quantized embeddings, thereby altering the code assignment to better align with the collaborative patterns (refer to Section~\ref{sec:exp_code_assignment_distribution} for empirical evidence).

\begin{figure}[t]
\setlength{\abovecaptionskip}{0cm}
\setlength{\belowcaptionskip}{-0.35cm}
\centering
\includegraphics[scale=0.68]{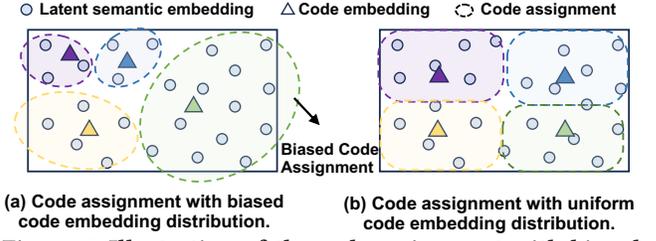}
\caption{Illustration of the code assignment with biased distribution and uniform distribution of code embeddings.}
\label{fig:diversity-method}
\end{figure}

\vspace{3pt}
\noindent$\bullet\quad$\textbf{Diversity regularization}.
To address the code generation bias, it is promising to learn a more uniform distribution of the code embeddings. 
As illustrated in Figure~\ref{fig:diversity-method}(a), a biased code embedding distribution is more likely to cause biased code assignment in the latent space, where some codes are assigned to more items. 
In contrast, a uniform code embedding distribution facilitates more balanced code assignment (Figure~\ref{fig:diversity-method}(b)). 
In light of these, we aim to improve the code embedding diversity for diverse code assignments. 

Specifically, for each codebook, we cluster the code embeddings into $K$ groups by constrained K-means~\cite{bradley2000constrained}. 
We then regularize the clustered code embeddings by a diversity loss, which is defined as
\begin{equation}
    \mathcal{L}_{\text{Div}}=-\frac{1}{B}\sum_{i=1}^{B}\frac{\exp (<\bm{e}_{c_l}^{i}, \bm{e}_{+}>)}{\sum_{j=1}^{N-1}\exp(<\bm{e}_{c_l}^{i},\bm{e}_j>)},
\end{equation}
where $\bm{e}_{c_l}^{i}$ is the nearest code embedding of item $i$, $\bm{e}_{+}$ denotes the code embedding of a randomly selected sample from the same cluster of code $c_l$, and $\bm{e}_{j\in{\{1, \dots, N\}}\setminus{c_{l}}}$ represents all code embeddings from the codebook except for $\bm{e}_{c_l}$. 
Intuitively,  as shown in Figure~\ref{fig:Main-Figure}, we pull the code embeddings from the same cluster closer and push the embeddings of codes from different clusters away from each other, thus encouraging diversity of the code embeddings and alleviating the code assignment bias issue (\cf Section~\ref{sec:exp_code_assignment_distribution} for empirical evidence). 

\vspace{3pt}
\noindent$\bullet\quad$\textbf{Overall loss}.
The training loss for LETTER is summarized as:
\begin{equation}\label{eq:tokenization_loss}
    \mathcal{L}_{\text{LETTER}}=\mathcal{L}_{\text{Sem}} + \alpha \mathcal{L}_{\text{CF}} + \beta \mathcal{L}_{\text{Div}}, 
\end{equation}
where $\alpha$ and $\beta$ are hyper-parameters to control the strength of collaborative regularization and diversity regularization, respectively. 

\subsection{Instantiation}\label{sec:instantiation}

To instantiate LETTER to LLM-based generative recommender models, we first train the tokenizer on the recommendation items via Eq. (\ref{eq:tokenization_loss}). 
And then the well-trained LETTER can be instantiated to tokenize items for the training and inference stages of LLMs. 

\subsubsection{\textbf{Training}}
The training of generative recommender models includes item tokenization and model optimization phases. 
Item tokenization index each item into the identifier $\Tilde{i}=[c_1, c_2, \dots, c_L]$ via the well-trained LETTER. 
We then translate the user interaction sequences based on item identifiers. Formally, we have $\mathcal{D}=\{(x,y)\}$, where $x=[\tilde{i}_1, \tilde{i}_2, \dots, \tilde{i}_M]$ denotes user's historically interacted items in chronological order, and $y=\tilde{i}_{M+1}$ refers to the identifier of user's next-interacted item. 

\vspace{2pt}
\noindent$\bullet\quad$\textbf{Ranking-guided generation loss}. 
Existing work usually optimizes LLMs on $\mathcal{D}=\{(x,y)\}$ by a generation loss for negative log-likelihood minimization. 
Despite the effectiveness of the generation loss in various domains, such generation loss might overlook the ranking optimization over all items, thus undermining the recommendation performance. 
In this light, we propose a ranking-guided generation loss, which alters the temperature in the traditional generalization loss to emphasize the penalty for hard-negative samples, thereby enhancing the ranking ability of generative recommender models. Formally, the loss is defined as:
\begin{equation}\label{eq:generative_loss}
\left\{
\begin{aligned}
   &\mathcal{L}_{\text{rank}}=-\sum_{t=1}^{|y|}\log P_{\theta}(y_t|y_{<t},x), \\
   &P_{\theta}(y_t|y_{<t},x)=\frac{\exp(p(y_t)/\tau)}{\sum_{v\in \mathcal{V}} \exp(p(v)/\tau)},
\end{aligned}
\right.
\end{equation}
where $\tau$ is the adjustable temperature to penalize the hard-negative samples, 
$y_t$ refers to the $t$-th token of $y$, $y_{<t}$ represents the token sequence preceding $y_t$, 
$\theta$ denotes the learnable parameters of the generative model, 
and $\mathcal{V}$ is the token vocabulary\footnote{The token vocabulary is extended by adding the codes from LETTER.}. 


\begin{proposition}
For a given ranking-guided generation loss $\mathcal{L}_{\text{rank}}$ and a parameter $\tau$, the following statements hold:
\begin{itemize}[leftmargin=*]
    \item Minimizing $\mathcal{L}_{\text{rank}}$ is equivalent to optimizing hard-negative items for users, where a smaller $\tau$ intensifies the penalty on hard negatives. 
    \item The minimization of $\mathcal{L}_{\text{rank}}$ is associated with the optimization of one-way partial AUC~\cite{shi2023theories}, which is strongly correlated with ranking metrics such as Recall and NDCG, ultimately leading to an improvement in the top-$K$ ranking ability.
\end{itemize}
\end{proposition}
The proof of this proposition is provided in Appendix~\ref{app:proof}.

\subsubsection{\textbf{Inference}}
To generate the next item, generative recommender models autoregressively generate the code sequence, which is formulated as $\hat{y}_{t} = \mathop{\arg\max}_{v\in \mathcal{V}} P_{\theta}(v|\hat{y}_{<t},x)$. 
To ensure generating valid identifiers, we follow~\cite{hua2023index} to employ constrained generation~\cite{de2020autoregressive} via Trie, which is a prefix tree~\cite{cormen2022introduction} supporting the model to find all strictly valid successor tokens in autoregressive generation. 

\section{Experiments}
\label{sec:exp}
In this section, we conduct extensive experiments on three real-world datasets to answer the following research questions:
\begin{itemize}[leftmargin=*]
    \item \textbf{RQ1:} How does our proposed LETTER perform compared to different kinds of identifiers?
    \item \textbf{RQ2:} How do the components of LETTER (\eg collaborative regularization and diversity regularization) affect the performance? 
    \item \textbf{RQ3:} How does LETTER perform under different settings (\ie identifier length, codebook size, strength of collaborative and diversity regularization, and temperature)? 
\end{itemize}

\subsection{Experimental Settings}

\subsubsection{\textbf{Datasets}} 
We use three real-world recommendation datasets from different domains: 
1) \textbf{Instruments} is from the Amazon review datasets~\cite{ni2019justifying}, which contains user interactions with rich music gears. 
2) \textbf{Beauty} contains user interactions with extensive beauty products from Amazon review datasets~\cite{ni2019justifying}. 
3) \textbf{Yelp}\footnote{\url{https://www.yelp.com/dataset}.} is a popular dataset comprising business interactions on the Yelp platform. 
We adopt the pre-processing techniques from previous work~\cite{kang2018self,rajput2023recommender}, discarding sparse users and items with interactions less than 5. 
We consider sequential recommendation setting\footnote{We focus on sequential recommendation setting due to its high practicality and significant potential in real-world applications.} and use \textit{leave-one-out} strategy~\cite{rajput2023recommender, zheng2023adapting} to split datasets. 
For training, we follow~\cite{hua2023index, zheng2023adapting} to restrict the number of items in a user's history to 20. 

\subsubsection{\textbf{Baselines}}
\
\
We compare LETTER with competitive baselines within two groups of work, traditional recommender models and LLM-based generative recommender models, with different types of identifiers (\ie ID, textual, and codebook-based identifiers): 1) \textbf{MF}~\cite{rendle2012bpr} decomposes the user-item interactions into the user embeddings and the item embeddings in the latent space. 
2) \textbf{Caser}~\cite{tang2018personalized} employs convolutional nerual networks to capture user's spatial and positional information.
3) \textbf{HGN}~\cite{ma2019hierarchical} utilizes graph neural networks to learn user and item representations for predicting user-item interactions.
4) \textbf{BERT4Rec}~\cite{sun2019bert4rec} leverages BERT's pre-trained language representations to capture semantic user-item relationships.
5) \textbf{LightGCN}~\cite{he2020lightgcn} is a lightweight graph convolutional network model focusing on high-order connections between users and items.
6) \textbf{SASRec}~\cite{kang2018self} employs self-attention mechanisms to capture long-term dependencies in user interaction history.
7) \textbf{BIGRec}~\cite{bao2023bi} is an LLM-based generative recommender model using textual identifiers, with each item represented by its title. 
8) \textbf{P5-TID}~\cite{hua2023index} uses items' title as textual identifiers for LLM-based generative recommender model.
9) \textbf{P5-SemID}~\cite{hua2023index} leverages item metadata (\eg attributes) to construct ID identifiers for LLM-based generative recommender model.
10) \textbf{P5-CID}~\cite{hua2023index} incorporates collaborative signals into the identifier for LLM-based generative recommender models through a spectral clustering tree derived from item co-appearance graphs.
11) \textbf{TIGER}~\cite{rajput2023recommender} introduces codebook-based identifiers via RQ-VAE, which quantizes semantic information into code sequence for LLM-based generative recommendation.  
12) \textbf{LC-Rec}~\cite{zheng2023adapting} uses codebook-based identifiers and auxiliary alignment tasks to better utilize knowledge in LLMs by connecting generated code sequences with natural language.

\vspace{2pt}
\noindent$\bullet\quad$\textbf{Evaluation settings.}
\
We use two metrics: top-$K$ Recall (R@$K$) and NDCG (N@$K$) with $K$ = 5 and 10, following~\cite{rajput2023recommender}.

\subsubsection{\textbf{Implementation Details}} 


We instantiate LETTER on two representative LLM-based generative recommender models, \ie TIGER~\cite{rajput2023recommender} and LC-Rec~\cite{zheng2023adapting}. 
As for TIGER, since the official implementations have not been released, we follow the paper for implementation. 
For LC-Rec, we perform the parameter-efficient fine-tuning technique 
LoRA~\cite{hu2021lora} to fine-tune LLaMA-7B~\cite{touvron2023llama}. 
To obtain item semantic embedding, we follow~\cite{zheng2023adapting} to adopt LLaMA-7B for encoding item content information. 
We obtain 32-dimensional item embedding from SASRec~\cite{kang2018self} for collaborative regularization and set cluster number $K$ at 10 for diversity regularization. 
All the experiments are conducted on 4 NVIDIA RTX A5000 GPUs.

For item tokenization, we use 4-level codebooks for RQ-VAE, where each codebook comprises 256 code embeddings with a dimension of 32. 
We train LETTER for 20k epochs via AdamW~\cite{loshchilov2017decoupled} optimizer with a learning rate of 1e-3 and a batch size of $1,024$. We follow~\cite{rajput2023recommender} to set $\mu$ as 0.25, and search $\alpha$ in a range of \{1e-1, 2e-2, 1e-2,1e-3\}, $\beta$ in a range of \{1e-2, 1e-3, 1e-4,1e-5\}. 
After LETTER training, we fine-tune TIGER and LC-Rec for convergence based on the validation performance with a learning rate in \{1e-3, 5e-4\} and \{1e-4, 2e-4, 3e-4\}. 



\begin{table*}[t]
\setlength{\abovecaptionskip}{0.05cm}
\setlength{\belowcaptionskip}{0.2cm}
\caption{Overall performance comparison between the baselines and LETTER instantiated on two competitive LLM-based generative recommender models on three datasets. The bold results highlight the better performance in comparing the backend models with and without LETTER.}
\setlength{\tabcolsep}{3mm}{
\resizebox{\textwidth}{!}{
\begin{tabular}{c|cccc|cccc|cccc}
\toprule
 &
  \multicolumn{4}{c|}{\multirow{-1}{*}{\textbf{Instruments}}} &
  \multicolumn{4}{c|}{\multirow{-1}{*}{\textbf{Beauty}}} &
  \multicolumn{4}{c}{\multirow{-1}{*}{\textbf{Yelp}}} \\
\multicolumn{1}{c|}{\textbf{Model}} &
  \multicolumn{1}{c}{\textbf{R@5}} &
  \multicolumn{1}{c}{\textbf{R@10}} &
  \multicolumn{1}{c}{\cellcolor[HTML]{FFFFFF}\textbf{N@5}} &
  \multicolumn{1}{c|}{\cellcolor[HTML]{FFFFFF}\textbf{N@10}} &
  \multicolumn{1}{c}{\textbf{R@5}} &
  \multicolumn{1}{c}{\textbf{R@10}} &
  \multicolumn{1}{c}{\cellcolor[HTML]{FFFFFF}\textbf{N@5}} &
  \multicolumn{1}{c|}{\cellcolor[HTML]{FFFFFF}\textbf{N@10}} &
  \multicolumn{1}{c}{\textbf{R@5}} &
  \multicolumn{1}{c}{\textbf{R@10}} &
  \multicolumn{1}{c}{\textbf{N@5}} &
  \multicolumn{1}{c}{\textbf{N@10}} \\   \midrule\midrule
\textbf{MF} &
  0.0479 &
  0.0735 &
  0.0330 &
  0.0412 &
  0.0294 &
  0.0474 &
  0.0145 &
  0.0191 &
  0.0220 &
  0.0381 &
  0.0138 &
  0.0190 \\
\textbf{Caser} &
  0.0543 &
  0.0710 &
  0.0355 &
  0.0409 &
  0.0205 &
  0.0347 &
  0.0131 &
  0.0176 &
  0.0150 &
  0.0263 &
  0.0099 &
  0.0134 \\
\textbf{HGN} &
  0.0813 &
  0.1048 &
  0.0668 &
  0.0774 &
  0.0325 &
  0.0512 &
  0.0206 &
  0.0266 &
  0.0186 &
  0.0326 &
  0.0115 &
  0.0159 \\
\textbf{Bert4Rec} &
  0.0671 &
  0.0822 &
  0.0560 &
  0.0608 &
  0.0203 &
  0.0347 &
  0.0124 &
  0.0170 &
  0.0186 &
  0.0291 &
  0.0115 &
  0.0159 \\
\textbf{LightGCN} &
  0.0794 &
  0.0100 &
  0.0662 &
  0.0728 &
  0.0305 &
  0.0511 &
  0.0194 &
  0.0260 &
  0.0248 &
  0.0407 &
  0.0156 &
  0.0207 \\
\textbf{SASRec} &
  0.0751 &
  0.0947 &
  0.0627 &
  0.0690 &
  0.0380 &
  0.0588 &
  0.0246 &
  0.0313 &
  0.0183 &
  0.0296 &
  0.0116 &
  0.0152 \\
\textbf{BigRec} &
  0.0513 &
  0.0576 &
  0.0470 &
  0.0491 &
  0.0243 &
  0.0299 &
  0.0181 &
  0.0198 &
  0.0154 &
0.0169 &
0.0137 &
0.0142 \\
\textbf{P5-TID} &
0.0000 &
  0.0001 &
  0.0000 &
  0.0000 &
0.0182 &
0.0432 &
0.0132 &
0.0254 &
  0.0184 &
0.0263 &
0.0130 &
0.0155 \\
\textbf{P5-SemID} &
0.0775 &
  0.0964 &
  0.0669 &
  0.0730 &
0.0393 &
0.0584 &
0.0273 &
0.0335 &
  0.0202 &
0.0324 &
0.0131 &
0.0170 \\
\textbf{P5-CID} &
  0.0809 &
  0.0987 &
  0.0695 &
  0.0751 &
0.0404 &
0.0597 &
0.0284 &
0.0347 &
  0.0219 &
0.0347 &
0.0140 &
0.0181 \\ \midrule
\textbf{TIGER} &
0.0870 &
0.1058 &
0.0737 &
0.0797 &
  0.0395 &
  0.0610 &
  0.0262 &
  0.0331 &
  0.0253 &
0.0407 &
0.0164 &
0.0213 \\
\textbf{LETTER-TIGER} &
\textbf{0.0909} &
\textbf{0.1122} &
\textbf{0.0763} &
\textbf{0.0831} &
\textbf{0.0431} &
  \textbf{0.0672} &
  \textbf{0.0286} &
  \textbf{0.0364} &
  \textbf{0.0277} &
\textbf{0.0426} &
\textbf{0.0184} &
\textbf{0.0231} \\ \midrule
\textbf{LC-Rec} &
0.0824 &
0.1006 &
0.0712 &
  0.0772 &
  0.0443 &
  0.0642 &
  0.0311 &
  0.0374 &
  0.0230 &
  0.0359 &
  0.0158 &
  0.0199 \\
\textbf{LETTER-LC-Rec} &
  \textbf{0.0913} &
  \textbf{0.1115} &
  \textbf{0.0789} &
  \textbf{0.0854} &
  \textbf{0.0505} &
  \textbf{0.0703} &
  \textbf{0.0355} &
  \textbf{0.0418} &
  \textbf{0.0255} &
  \textbf{0.0393} &
  \textbf{0.0168} &
  \textbf{0.0211} \\ \bottomrule
\end{tabular}
}}
\label{tab:MainTable}
\end{table*}

\subsection{Overall Performance (RQ1)}


We evaluate LETTER on two SOTA LLM-based generative recommender models. The performances of LETTER instantiated on TIGER (LETTER-TIGER) and LC-REC (LETTER-LC-Rec) are reported in Table~\ref{tab:MainTable}, from which we have the following observations:

\begin{itemize}[leftmargin=*]

    \item Among the LLM-based models with ID identifiers (P5-CID and P5-SemID), P5-CID usually yields better performance than P5-SemID over the three datasets. 
    This is because 
    P5-SemID assigns item identifiers based on the item categories, \eg ``string'' and ``keyboard'' for instrument products, which might 1) fail to capture fine-grained semantic information and 2) hinder the learning of collaborative behavior in LLMs' fine-tuning stage due to the misalignment between semantic and collaborative signals (\cf Figure~\ref{fig:cf_misalignment} in Section~\ref{sec:task}). 
    In contrast, P5-CID leverages collaborative signals from the item co-appearance graph to assign identifiers, which is beneficial for LLM-based recommender models to capture collaborative patterns in user behaviors. 


    \item Among all LLM-based baselines with different kinds of identifiers, models with codebook-based identifiers (TIGER and LC-Rec) outperform models with ID identifiers (P5-SemID and P5-CID) and that with textual identifiers (BIGRec and P5-TID) in most cases. 
    This is attributed to the incorporation of hierarchical semantics with different granularities via RQ-VAE, which distinguishes between similar items more effectively by grasping fine-grained details. 
    Meanwhile, the inferior performance of BIGRec and P5-TID with textual identifiers is probably due to the potential misalignment between items' similar semantics and dissimilar interactions, thus hurting the learning of collaborative signals for accurate recommendation. 
    
    \item LETTER consistently exhibits significant improvements for the backend models (TIGER and LC-Rec) across three datasets, validating the effectiveness of our method. 
    The superior performance is attributed to: 
    1) the CF integration during code assignment, which aligns the collaborative signals and the semantic embeddings, thereby encouraging items with similar interactions or semantics to have similar code sequence and addressing the misalignment issue between semantics and collaborative signals (refer to Section~\ref{sec:exp_cf_information} for further analysis of CF integration). 
    2) The improved diversity of token assignments, which is achieved by regularizing the representation space of code embeddings, thereby alleviating item generation bias caused by code assignment bias (refer to Section~\ref{sec:exp_code_assignment_distribution} for detailed analysis). 
\end{itemize}

\subsection{In-depth Analysis}

\subsubsection{\textbf{Ablation Study (RQ2)}}

To thoroughly investigate the effect of each regularization, we compare the following five variants of LETTER on TIGER: 
(0) we solely employ semantic regularization for item tokenization, which is equivalent to TIGER. 
(1) We incorporate both semantic and collaborative regularization for item tokenization, denoted as ``TIGER w/ c. r.''. 
(2) We involve both semantic and diversity regularization for item tokenization, referred to as ``TIGER w/ d. r.''.  
(3) We utilize semantic, collaborative, and diversity regularization for item tokenization and train TIGER with original generation loss, denoted as ``(1) w/ d. r.''.  
(4) We employ all regularizations for item tokenization and apply ranking-guided generation loss, \ie LETTER-TIGER. 

\begin{table}[]
\setlength{\abovecaptionskip}{0cm}
\setlength{\belowcaptionskip}{0cm}
\caption{Ablation study of LETTER-TIGER.}
\setlength{\tabcolsep}{3.8mm}{
\resizebox{0.48\textwidth}{!}{
\begin{tabular}{l|cc|cc}
\toprule
\multicolumn{1}{c|}{\textbf{}} & \multicolumn{2}{c|}{\textbf{Instruments}}                       & \multicolumn{2}{c}{\textbf{Beauty}}                             \\
\multicolumn{1}{c|}{\textbf{Variants}} & \textbf{R@10} & \cellcolor[HTML]{FFFFFF}\textbf{N@10} & \textbf{R@10} & \textbf{N@10} \\ \hline
\textbf{(0): TIGER}            & 0.1058                          & 0.0797                         & 0.0610                         & 0.0331                         \\
\textbf{(1): TIGER w/ c. r.}    & 0.1078                         & 0.0810                         & 0.0660                          & 0.0351                         \\
\textbf{(2): TIGER w/ d. r.}   & 0.1075                         & 0.0809                         & 0.0618 & 0.0335 \\
\textbf{(3): (1) w/ d. r.}     & 0.1092 & 0.0819 & 0.0672                         & 0.0357                         \\
\textbf{(4): LETTER-TIGER}     & \textbf{0.1122}                & \textbf{0.0831}                & \textbf{0.0672}                & \textbf{0.0364}                \\ \hline
\end{tabular}
}}
\label{tab:ablation}
\vspace{-0.3cm}
\end{table}

Results on Instruments and Beauty are depicted in Table~\ref{tab:ablation}. 
We omit the results on Yelp with similar observations to save space. 
From the table, we have several key observations: 
1) incorporating either collaborative or diversity regularization can improve the performance of TIGER, which validates the effectiveness of injecting collaborative signals into code assignment and improving the diversity of code embeddings, respectively. 
2) Optimizing item tokenization with all three regularizations will further improve the performance (superior performance of (3) than (0-2)), indicating the effectiveness of considering both semantics and collaborative information in the code assignment with improved diversity. 
3) LETTER-TIGER achieves the best performance, which leverages ranking-guided generation loss. This shows the effectiveness of suppressing hard-negative samples by altering temperature to strengthen the ranking ability. 


\subsubsection{\textbf{Code Assignment Distribution (RQ2)}}\label{sec:exp_code_assignment_distribution} 

\begin{figure*}[t]
\setlength{\abovecaptionskip}{0.05cm}
\setlength{\belowcaptionskip}{-0.2cm}
\centering
\includegraphics[scale=0.5]{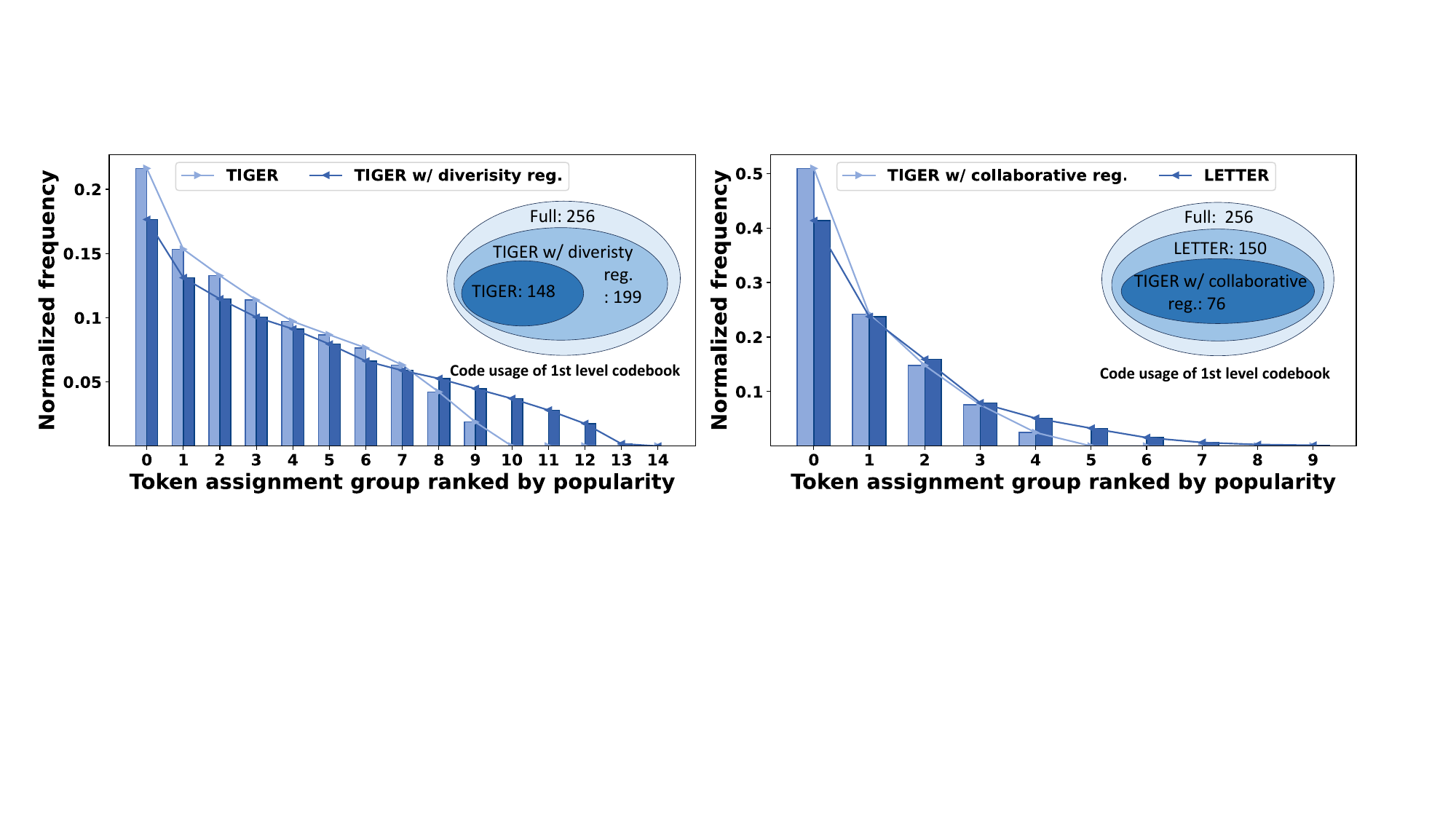}
\caption{Illustration of smoother code assignment distribution on Instruments. 
The left figure compares the code assignment between TIGER and TIGER with diversity regularization and the right figure compares between TIGER with collaborative regularization and LETTER. ``reg'' denotes ``regularization''.}
\label{fig:exp-code-assignment-distribution}
\end{figure*}



We further investigate whether diversity regularization can effectively mitigate the code assignment bias in item tokenization. 
We compare the distributions of the first code of identifier tokenized by 1) TIGER with (w/) and without diversity regularization (left figure in Figure~\ref{fig:exp-code-assignment-distribution}), and 
2) TIGER equipped with collaborative regularization with and without diversity regularization (right figure in Figure~\ref{fig:exp-code-assignment-distribution}), respectively. 
Specifically, for each tokenization approach, we first tokenize items with the well-trained tokenizer. We then sort the first code of the identifier according to the frequency and divide the first codes into groups with each of the groups containing 15 codes. 

From Figure~\ref{fig:exp-code-assignment-distribution}, we observe that: 
1) incorporating diversity regularization effectively promotes a more uniform distribution of the first code, thus alleviating code assignment bias and potentially mitigating the item generation bias of generative recommender models. 
2) The incorporation of diversity regularization significantly raises the utilization rate of codes in the first-level codebook, thereby improving diversity in code assignment. 
3) Although adding collaborative regularization decreases the code utilization of TIGER (from 148 to 76), integrating diversity regularization inversely compensates for the drop in code utilization. 
As such, LETTER can capture collaborative signals and maintain a high code utilization, simultaneously meeting the multiple criteria of an ideal identifier.



\subsubsection{\textbf{Code Embedding Distribution (RQ2)}}\label{sec:code_emb_distribution}

\begin{figure}[t]
\setlength{\abovecaptionskip}{0.05cm}
\setlength{\belowcaptionskip}{-0.1cm}
\centering
\includegraphics[scale=0.34]{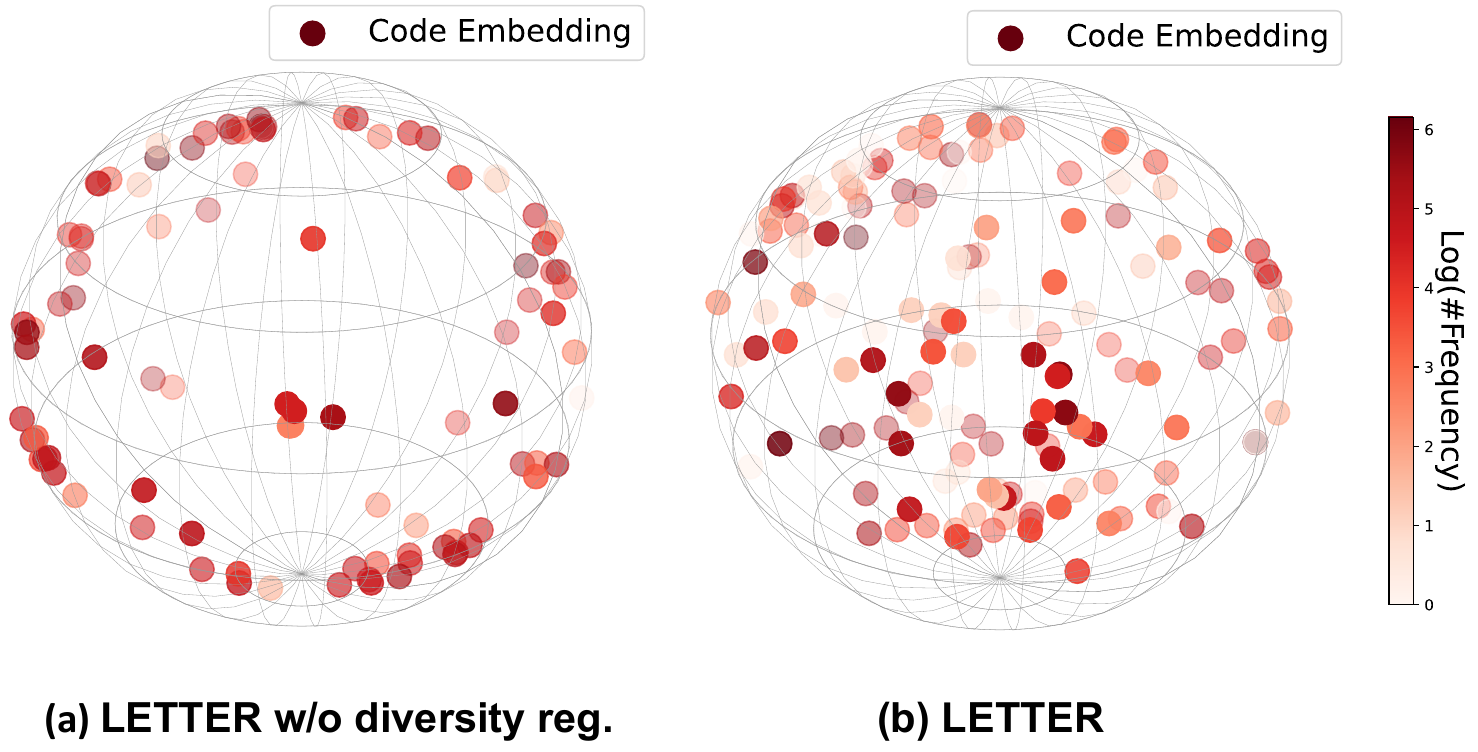}
\caption{Code embedding distribution of LETTER with and without diversity regularization on Instruments.}
\label{fig:exp-code-embedding-distribution}
\vspace{-0.3cm}
\end{figure}


To analyze whether diversity regularization can mitigate the biased distribution of code embeddings, we visualize the code embeddings of LETTER and LETTER without (w/o) diversity regularization. 
we first obtain the code embeddings from the first-level codebook and then perform PCA to reduce the high-dimensional embeddings into the 3-dimensional space for visualization. 
As shown in Figure~\ref{fig:exp-code-embedding-distribution}, we present the 3-dimensional code embeddings in a spherical plot, where each point represents a code embedding vector, and darker colors indicate that the code is assigned to more items. 
By comparing (a) LETTER w/o diversity regularization and (b) LETTER, we can find that the first code embeddings of LETTER are more evenly distributed in the embedding representation space compared to LETTER w/o diversity regularization. 
This validates the effectiveness of diversity regularization to achieve a more diverse distribution of code embeddings in the representation space, fundamentally alleviating the problem of biased code assignment in Figure~\ref{fig:diversity-method}(a).

\subsubsection{\textbf{Investigation on Collaborative Signals in Identifiers (RQ2).}} \label{sec:exp_cf_information}




\begin{table}[]
\setlength{\abovecaptionskip}{0cm}
\setlength{\belowcaptionskip}{-0.5cm}
\caption{Ranking performance with quantized embeddings.}
\setlength{\tabcolsep}{3mm}{
\resizebox{0.47\textwidth}{!}{
\begin{tabular}{cc|cccc}
\toprule
\textbf{Dataset} &
  \textbf{Model} &
  \multicolumn{1}{c}{\textbf{R@5}} &
  \multicolumn{1}{c}{\textbf{R@10}} &
  \multicolumn{1}{c}{\textbf{N@5}} &
  \multicolumn{1}{c}{\textbf{N@10}} \\ \midrule
                                  & \textbf{TIGER}     & 0.0050          & 0.0150          & 0.0024          & 0.0049          \\ 
                                  
\multirow{-2}{*}{\textbf{Instruments}} & \textbf{LETTER}   &
  \textbf{0.0080} &
  \textbf{0.0159} &
  \textbf{0.0038} &
  \textbf{0.0058} \\         \hline            
\multicolumn{1}{l}{}              &  \textbf{TIGER}                    & 0.0128          & 0.0213          & 0.0064          & 0.0085          \\ 
\multicolumn{1}{c}{\multirow{-2}{*}{\textbf{Beauty}}} &
\textbf{LETTER} & \textbf{0.0175} & \textbf{0.0343} & \textbf{0.0076} & \textbf{0.0118} \\  \bottomrule
\end{tabular}
}}
\label{tab:exp-ranking}
\end{table}


\begin{table}[]
\setlength{\abovecaptionskip}{0cm}
\setlength{\belowcaptionskip}{-0.2cm}
\caption{Code similarity between items with similar collaborative signals.}
\setlength{\tabcolsep}{10mm}{
\resizebox{0.45\textwidth}{!}{
\begin{tabular}{c|clllclll}
\toprule
      & \multicolumn{4}{l}{\textbf{Instruments}}  & \multicolumn{4}{l}{\textbf{Beauty}} \\ \midrule
\textbf{TIGER} & \multicolumn{4}{c}{0.0849}          & \multicolumn{4}{c}{0.1135}    \\
\textbf{LETTER}  & \multicolumn{4}{c}{\textbf{0.2760}} & \multicolumn{4}{c}{\textbf{0.3312}}    \\ \bottomrule
\end{tabular}
}}
\label{tab:exp-similarity}
\vspace{-0.3cm}
\end{table}

To verify whether LETTER encodes collaborative signals into identifiers as expected, we design two experiments for analysis: 

\noindent$\bullet\quad$\textbf{Ranking experiment}. 
We aim to assess the ranking performance of LETTER by utilizing the quantized embeddings of items for interaction prediction. 
Specifically, we first obtain the item's quantized embedding $\hat{\bm{z}}$ from the well-trained tokenizer. 
We then replace the item embeddings of the well-trained traditional CF model (\ie SASRec) with the quantized embeddings for interaction prediction. 
Intuitively, identifiers that effectively capture collaborative signals will lead to a better ranking performance. 
From Table~\ref{tab:exp-ranking}, we can observe that LETTER outperforms TIGER by a large margin, indicating the effective incorporation of collaborative signals. 

\begin{figure*}[t]
\setlength{\abovecaptionskip}{0cm}
\setlength{\belowcaptionskip}{-0.2cm}
\centering
\includegraphics[scale=0.58]{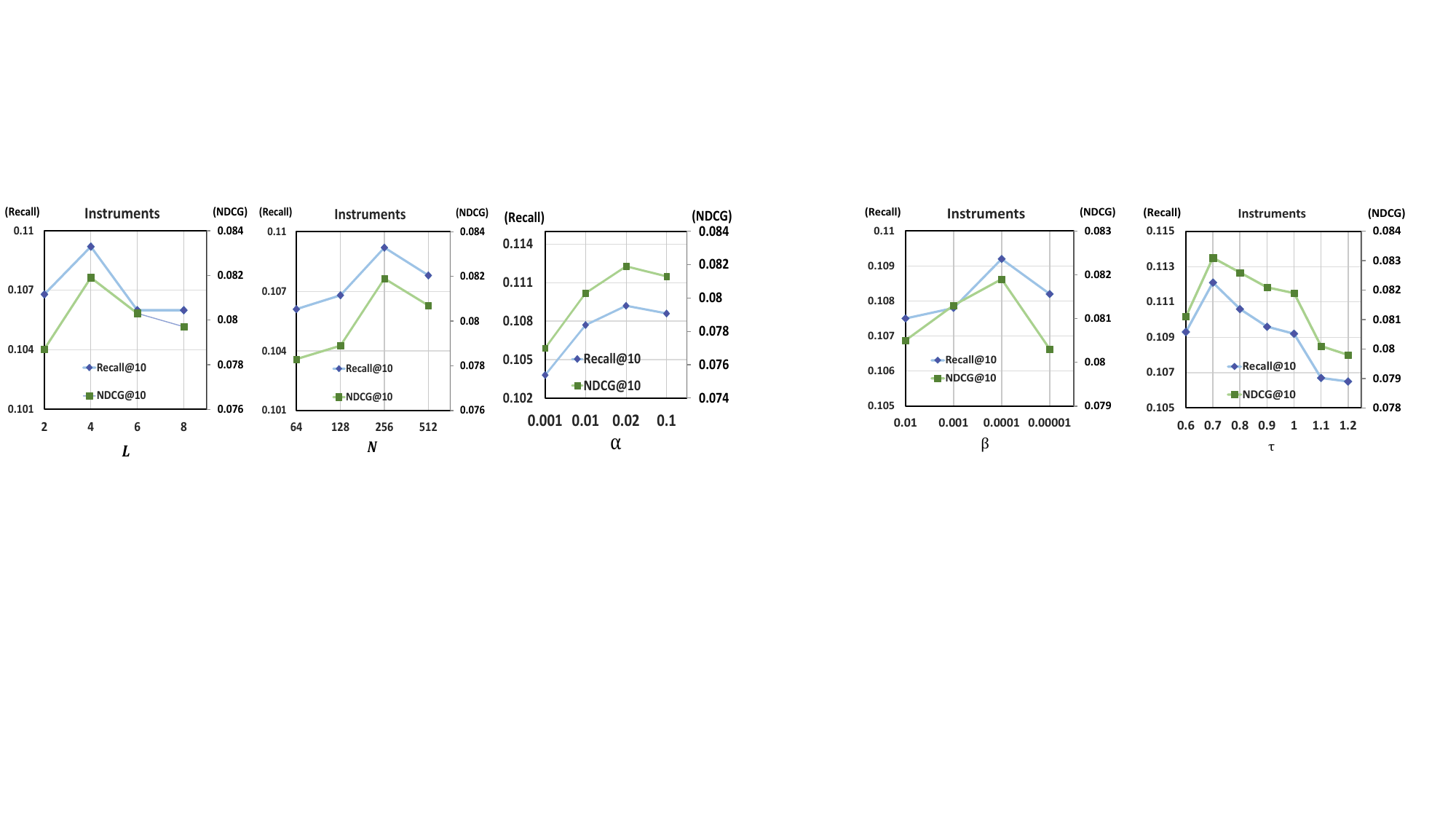}
\caption{Performance of LETTER-TIGER over different hyper-parameters on Instruments. }
\label{fig:exp-hyper-parameters}
\end{figure*}

\vspace{1pt}
\noindent$\bullet\quad$\textbf{Similarity experiment}. 
Our target is to verify whether the items with similar collaborative signals exhibit similar identifiers. 
We first pair every item with its most similar item based on the similarity from pre-trained CF embeddings. 
Next, we assess the similarity of the code sequence between the two items within the item pairs. 
Specifically, we measure the overlap degree between the code sequences of the two items and report the averaged results of all items in Table~\ref{tab:exp-similarity}. 
We can find that LETTER achieves more similar code sequences for the items with similar collaborative signals than TIGER, thereby alleviating the misalignment issue between semantics and the collaborative similarity. 



\subsubsection{\textbf{Hyper-Parameter Analysis (RQ3)}} 
We further investigate several important hyper-parameters introduced in LETTER to facilitate future applications. 

\noindent$\bullet\quad$\textbf{Identifier length $L$.} 
We vary $L$ from 2 to 8 and present the results in Figure~\ref{fig:exp-hyper-parameters}. 
It is observed that 
1) the performance improves as we increase $L$ from 2 to 4. 
Because overly short identifiers may lose some fine-grained information, resulting in weaker expressiveness of idetentifiers. 
2) Continuously increasing the identifier length from 4 to 8 will inversely degrade the performance. This is reasonable because autoregressive generation suffers from error accumulation~\cite{ren2024non}. 
Since accurate item generation requires the correctness of all generated codes in the sequence, generating longer identifiers accurately is more challenging than shorter identifiers. 

\noindent$\bullet\quad$\textbf{Codebook size $N$.}
We evaluate LETTER with different codebook sizes $N=$ 64, 128, 256, and 512. 
Results in Figure~\ref{fig:exp-hyper-parameters} show that 
1) Gradually increasing $N$ tends to perform better. 
The worse performance of small-sized codebooks may be caused by the limited diversity of code selection, thus failing to distinguish items effectively. 
However, 2) blindly expanding $N$ might inversely hurt the performance, 
where larger codebook size might be more susceptible to noise in the item's semantic information, potentially leading to the over-fitting of some meaningless semantics. 


\noindent$\bullet\quad$\textbf{Strength of collaborative regularization $\alpha$.}
We vary $\alpha$ from 0.001 to 0.1 and present the results in Figure~\ref{fig:exp-hyper-parameters}. We can find that 
1) as $\alpha$ continuously increases, the overall trend of performance is observed to be generally improved.
This is reasonable because larger $\alpha$ implies a stronger injection of collaborative patterns, but an overly large $\alpha$ may instead interfere with semantic regularization.
2) Empirically, we recommend setting $\alpha=0.02$ because it probably achieves a proper balance between semantic and collaborative regularization, thus leading to the best performance.

\noindent$\bullet\quad$\textbf{Strength of diversity regularization $\beta$.}
We examine the influence of different diversity regularization strengths by adjusting $\beta$ from 0.00001 to 0.01. 
The results are in Figure~\ref{fig:exp-hyper-parameters} and we observe that 
applying diversity regularization with a slight strength is sufficient to enhance the diversity of code assignment (significant improvements from $\beta=0.00001$ to $\beta=0.0001$). 
In contrast, an excessive amount of diversity signal may interfere with the tokenizer's integration of semantic and collaborative signals. 

\noindent$\bullet\quad$\textbf{Cluster $K$.} 
Based on the results with various values of $K$ ranging from 5 to 20, as illustrated in Figure~\ref{fig:exp-hyper-parameters}, 
we can observe that decreasing or increasing $K$ from 10, the performance tends to decline. 
This is understandable because overly large clusters contain too many code embeddings, causing embeddings within the same cluster to be insufficiently close, whereas overly small clusters have relatively few code embeddings, leading to embeddings within the same cluster being overly close.

\noindent$\bullet\quad$\textbf{Temperature $\tau$.} 
From the results of LETTER-TIGER with multiple values of $\tau$ ranging from 0.6 to 1.2 as shown in Figure~\ref{fig:exp-hyper-parameters}, 
we can find that 
decreasing $\tau$ from 1.2 to 0.7, the performance tends to decline. 
This is reasonable as decreasing temperature casts more emphasis on the penalty for hard negatives, where the ranking ability is strengthened, thus improving the recommendation performance. 
Nevertheless, we should carefully choose $\tau$ since a too-small $\tau$ might suppress the possibility of hard negative samples being considered as positive samples for other users.

\section{Related Work}
\label{sec:related_work}












$\bullet\quad$\textbf{LLMs for Generative Recommendation.}
LLMs have shown promising prospects for generative recommendation~\cite{li2024survey, lin2024bridging, gong2023unified,lin2024data,Sun2023LearningTT}, where item tokenization is a crucial problem. 
Existing work on item tokenization primarily investigate three types of identifiers for item indexing: ID identifiers~\cite{hua2023index, geng2022recommendation, wang2024enhanced,chu2023leveraging}, textual identifiers~\cite{bao2023bi,zhang2021language,zhang2023recommendation,liao2023llara, dai2023uncovering,li2023generative}, codebook-based identifiers~\cite{wang2024enhanced,rajput2023recommender,zheng2023adapting}. 
In early stages, randomly assigned IDs struggled to effectively encode semantic information and collaborative patterns of items.
So ID identifiers such as SemID~\cite{hua2023index} and CID~\cite{hua2023index} use item semantic information and collaborative signals, respectively, to build tree-like structures for generating identifiers. However, such a fixed and unlearnable structure is hard to efficiently and effectively represent item similarity and adapt new items.
Alternatively, textual identifiers leverage detailed descriptive information of items as identifiers. Nevertheless, these methods struggle to hierarchically encode semantic information or integrate collaborative signals.
Codebook-based identifiers leverage codebooks and attempt to integrate semantic information and collaborative signals during training process. However, they overlook the misalignment introduced by incorporating collaborative signals into fixed code sequences, as well as the code assignment bias.
Specifically, LC-Rec uses the Sinkhorn-Knopp Algorithm to fairly consider intra-layer codes through iterative normalization, missing the essence of code assignment: the distribution of code embeddings.
In this work, we propose a superior identifier that integrates hierarchical semantic with collaborative signals and code assignment diversity.

\vspace{3pt}
\noindent$\bullet\quad$\textbf{LLMs for Discriminative Recommendation.}
LLMs are used not only for generating recommendation but also extensively in different discriminative recommendation tasks. Two main approaches are employed in utilizing LLMs for discriminative recommendation. 1) LLM-enhanced traditional recommender. Many studies leverage LLMs for data augmentation~\cite{wei2024llmrec, xi2023towards,liu2024once} and representation~\cite{qiu2021u,wu2022userbert,ren2023representation}. 
2) LLM-based recommender, which utilizes LLMs as recommender models directly. 
Some methods~\cite{li2023e4srec, zhang2023collm,lin2023clickprompt} incorporate a projection layer to calculate the matching score between users and items for ranking. 
Besides, another line~\cite{bao2023tallrec,lin2023rella,prakash2023cr} leverages LLMs for click-through rate (CTR) prediction, combining user and item features to determine whether the user will like this item. 
In this work, we mainly focus on the identifier design for generative recommendation, and thus we do not delve into in-depth discussion and comparison with discriminative methods. 
\section{Conclusion and Future Work}
\label{sec:conclusion}

In this study, we undertook a thorough analysis of the optimal features required for item tokenization in generative recommendation. Subsequently, we introduced LETTER, a learnable tokenizer incorporating three forms of regularization aimed at capturing hierarchical semantics, collaborative signals, and code assignment diversity within item identifiers. We applied LETTER to two prominent generative recommender models and employed a ranking-guided generation loss to enhance their ranking performance. Extensive experiments substantiate the superiority of LETTER in item tokenization for generative recommendation models. 

This work highlights the importance of item tokenization on generative recommendation, leaving some promising direction for future exploration.
1) Tokenization with rich user behaviors to enable generative recommender models to infer user preference from multiple user actions.
2) LETTER has the potential to tokenize cross-domain items for open-ended recommendation, enabling generative recommender models to leverage multi-domain user behaviors and items for user preference reasoning and next-item recommendation. 
3) It is promising to combine user instructions in natural language with user interaction history tokenized by LETTER for personalized recommendations, achieving collaborative reasoning with complex natural language instructions and item tokens in the space of generative recommender models.
\section{Appendix}\label{app:proof}

Our proof is primarily divided into hard negative mining and relation with ranking two parts. 

\noindent$\bullet\quad$\textbf{Hard Negative Mining.}
This subsection primarily proofs two objects: Firstly, the ranking-guided loss $\mathcal{L}_{Rank}$ can mine hard negative samples. Secondly, as the value of $\tau$ decreases, there is an increased focus on hard negative samples.
We first prove hard negative mining ability of $\mathcal{L}_{Rank}$:
\begin{equation}
\footnotesize
\begin{aligned}
    \nabla \mathcal{L}_{rank} & = - \nabla \sum^{|y|}_{t=1} log P_{\theta}(y_{t}|y_{<t},x) 
    = - \sum^{|y|}_{t=1} \nabla log \frac{\exp(p(y_t)/\tau)}{\sum_{v\in \mathcal{V}} \exp(p(v)/\tau)}.
\end{aligned}
\end{equation}

We shall demonstrate that when t-1 tokens are determined, our loss is associated with hard negative mining on the t-th token.
\begin{equation}
\footnotesize
\begin{aligned}
    & - \nabla log  \frac{\exp(p(y_t)/\tau)}{\sum_{v\in \mathcal{V}} \exp(p(v)/\tau)} = 
 - \frac{1}{\tau} \nabla p(y_t) +  \nabla log \sum_{v\in \mathcal{V}} \exp(p(v)/\tau) \\
\end{aligned}
\end{equation}
\vspace{-10pt}
\begin{equation}
\small
\begin{aligned}
    = \frac{\sum\limits_{\substack{v\in \mathcal{V}, v\neq p_{t}}} \exp(p(v)/\tau)}{\tau \sum\limits_{v\in \mathcal{V}} \exp(\frac{p(v)}{\tau})} [\sum_{\substack{v\in \mathcal{V},\\ v\neq p_{t}}}\left( \frac{\exp(\frac{p(v)}{\tau})}{\sum_{\substack{v\in \mathcal{V}, \\ v\neq p_{t}}} \exp(\frac{p(v)}{\tau}) } \nabla p(v) \right)
      - \nabla p(y_t)].
    \label{eq:gradient}
\end{aligned}
\end{equation}
Next, we analyze the gradient $p(v)$ that associated with hard negative mining. Noting that this gradient is equivalent to minimizing the hard negative mining loss $\mathcal{L}_H:= \sum_{v\in \mathcal{V}, v\neq p_{t}} w(p(v), \tau) p(v),$ 
where $w(p(v), \tau) := sg\left[ \frac{\exp(p(v)/\tau)}{\sum_{v\in \mathcal{V}, v\neq p_{t}} \exp(p(v)/\tau) } \right]$,  $sg[\cdot]$ denotes stop-gradient operation. 
We can observe that:
    $w(p(v), \tau) \propto p(v)$
This means the harder the negative sample is, the larger $w(p(v), \tau)$ is, demonstrating the ability of hard negative sample mining of our ranking-guided loss $\mathcal{L}_{Rank}$.

Next, we elucidate that with the diminishing value of 
$\tau$, there is a heightened emphasis on hard negative samples. For a given negative sample $v'$, 
if 
    $p(v') > \frac{\sum_{v\in \mathcal{V}, v\neq p_{t}} \exp(p(v)/\tau) p(v)}{\sum_{v\in \mathcal{V}, v\neq p_{t}} \exp(p(v)/\tau)}$,
we have
$\partial_\tau w(p(v'), \tau) < 0$,
showing adjusting the value of $\tau$ in the ranking-guided loss mechanism increases the weights of more difficult negative samples with higher $p(v)$ and decreases the weights of simpler negative samples.
\\\noindent$\bullet\quad$\textbf{Relation with Ranking.}
Subsequently, we reveal the correlation between our loss and ranking metric. 
Firstly, let us observe the gradient in Eq.\ref{eq:gradient}, we have:
\begin{equation}
\small
\begin{aligned}
\label{eq:rho_dro}
    &\sum_{v\in \mathcal{V}, v\neq p_{t}}\left( \frac{\exp(p(v)\tau)}{\sum_{\substack{v\in \mathcal{V}, \\ v\neq p_{t}}} \exp(p(v)/\tau) } \nabla p(v) \right) - \nabla p(y_t)
    \\&= \nabla \underbrace{\left\{ \tau \log \left[ \sum_{v\in \mathcal{V}, v\neq p_{t}} \exp(L_{\theta}(v,y_t)/\tau) \right] + \tau\rho \right\}}_{\mathcal{L}_{\text{DRO}}^{\rho,\tau}},
\end{aligned}
\end{equation}

where we define: $L_{\theta}(v,y_t) := p(v) - p(y_t)$, $\theta$ is the notation for the parameters in the recommender system,
Eq.~\ref{eq:rho_dro} for $\forall \rho \in \mathcal{R}$ established. This means $\mathcal{L}_{rank}$ is a surrogate loss for $\mathcal{L}_{\text{DRO}}^{\rho,\tau}$. Then we analyze $\mathcal{L}_{\text{DRO}}^{\rho,\tau}$:
\begin{equation}
\small
\begin{aligned}
    \mathcal{L}_{\text{DRO}}^{\rho,\tau} \geq \underbrace{\inf_{\tau > 0} \left\{ \tau \log \left[ \sum_{v\in \mathcal{V}, v\neq p_{t}} \exp(L_{\theta}(v,y_t)/\tau) \right] + \tau\rho \right\}}_{\mathcal{L}_{DRO}^{\rho}},
    \label{eq:upper_bound}
\end{aligned}
\end{equation}
which demonstrates that $\mathcal{L}_{\text{DRO}}^{\rho,\tau}$ is a upper bound for every $\mathcal{L}_{DRO}^{\rho}$. Similar to \cite{KLDRO}, minimizing $\mathcal{L}_{DRO}^{\rho}$ is the close form of the following optimization problem:
\begin{equation}
\small
\begin{aligned}
\label{eq:close_form_kl}
    \min_{\theta} \max_{Q} E_{Q}\left[L_{\theta}(v,y_{t})\right],
    \text{Subject to } D_{KL}(Q||P_0) \leq \rho,
\end{aligned}
\end{equation}
where $P_0$ is the uniform distribution of negative samples, $D_{KL}$ is the KL divergence.
\begin{lemma}
    \label{lemma1} (Brought from Theorem 2 in \cite{shi2023theories}) optimize problem \ref{eq:close_form_kl} is associated with optimizing one-way partial AUC (OPAUC). Furthermore, $\mathcal{L}_{DRO}^{\rho}$ is a surrogate version of OPAUC($e^{-\rho}$) object. The formation of OPAUC($\beta$), $\beta \in [0,1]$ can be presented as following \cite{zhu2022auc}:
\begin{equation}
\small
\begin{aligned}
        OPAUC(\beta) = \int_0^{\beta} TPR[FPR^{-1}(s)] ds,
    \end{aligned}
    \end{equation}
    where $TPR$ is true positive rate and $FPR$ is false positive rate.
\end{lemma}


Subsequently, we shall demonstrate the relation between $OPAUC(\beta)$ and top-K ranking metric, finally showing the relation between ranking-guided loss and ranking metrics:
\begin{theorem}

\label{thm:rank}
    Considering only one positive sample, OPAUC has strong correlation with top-K ranking metric:
    \begin{equation}
    \small
    \begin{aligned}
        &Recall@K = \mathbb{I}(OPAUC(K/(|\mathcal{V}| - 1)) > 0), \\
        &NDCG@K = \frac{\mathbb{I}(OPAUC(K/(|\mathcal{V}| - 1)) > 0)}{\log_2((K-1)(1-OPAUC(K/(|\mathcal{V}| - 1))) + 2)}, 
    \end{aligned}
    \end{equation}
    where $\mathbb{I}(\cdot)$ is the indicator function, for $\forall K>1$ exists. 
\end{theorem}
\begin{proof}
Since we only have one positive samples, we can derive the exact rank position of the positive sample based on OPAUC. We use $r_p$ to denote the ranking position of the positive sample.
Set $\beta = K/(|V| - 1)$, we have:
    $OPAUC(K/(|V| - 1)) = (K-r)/(K-1) \mathbb{I}(r_p < K).$ 
If $r_p < K$, we have:
    $r_p = (K-1)(1-OPAUC(K/(|V| - 1))).$
At the same time, the value of $OPAUC(K/(|V| - 1))$ can also be used to determine whether $r_p < K$:
    $\mathbb{I}(OPAUC(K/(|V| - 1) > 0) = \mathbb{I}(r_p < K)$.
Since we have:
    $Recall@K = \mathbb{I}(r_p<K
),NDCG@K = \mathbb{I}(r_p<K) /\log_2(r_p + 1)$ substitute the results of $r_p$ in to $Recall@K$ and $NDCG@K$ to finish the proof.
\end{proof}

\clearpage

{
\balance
\bibliographystyle{plainnat}
\bibliography{reference}
}



\end{document}